\documentclass[journal]{IEEEtran}
\usepackage{epsfig,epsf,epstopdf}
\usepackage{hyperref}
\usepackage[dvipsnames]{xcolor}
\usepackage{booktabs} 
\usepackage{multicol}
\usepackage{subcaption}
\usepackage{amsthm}

\makeatletter
\renewenvironment{proof}[1][\proofname]{\par
  \pushQED{\qed}%
  \normalfont \topsep6\p@\@plus6\p@\relax
  \trivlist
  \item[\hskip\labelsep
        {\bfseries\itshape
    #1:}\@addpunct{}]\ignorespaces
}{%
  \popQED\endtrivlist\@endpefalse
}
\makeatother

\newtheoremstyle{boldhead}
  {} 
  {} 
  {\normalfont} 
  {} 
  {\bfseries\itshape} 
  {.} 
  { } 
  {\thmname{#1}\thmnumber{ #2}\thmnote{ \normalfont(#3)}} 

\theoremstyle{boldhead}
\newtheorem{theorem}{Theorem}

\newtheorem{assumption}{Assumption}
\newtheorem{definition}{Definition}
\newtheorem{remark}{Remark}

\graphicspath{ {./Figs/} }
\usepackage[cmex10]{amsmath}
\usepackage{amsfonts}
\usepackage{textcomp}
\usepackage{amssymb}

\usepackage[noend]{algpseudocode}
\makeatletter
\def\BState{\State\hskip-\ALG@thistlm}
\makeatother
\usepackage{cite}
\usepackage{paralist}
\usepackage{color}

\definecolor{blue}{rgb}{0, 0.1, 0.7}
\usepackage{mathrsfs}
\usepackage{multirow}
\usepackage{graphicx}
\usepackage{caption}
\usepackage[font=small,justification=justified,singlelinecheck=false]{caption}

\ifCLASSINFOpdf
 
\else
 
\fi

\usepackage{xcolor}
\begin{document}

\title{A Unified Framework for Attack-Resilient CLF-CBF Quadratic Programs for Nonlinear Control-Affine Systems}
\author{Mohamadamin Rajabinezhad, and Shan Zuo
\thanks{Mohamadamin Rajabinezhad, and Shan Zuo are with the Department of Electrical and Computer Engineering, University of Connecticut, Storrs, CT 06269, USA. (Emails:mohamadamin.rajabinezhad@uconn.edu; shan.zuo@uconn.edu)(This work is under review for possible publication.)}
}

\maketitle

\begin{abstract} 
This letter introduces attack-resilient Control Lyapunov Functions
(AR-CLFs) and attack-resilient Control Barrier Functions (AR-CBFs) for nonlinear control-affine systems subject to control-input false data injection attacks (FDIA) satisfying an at-most-exponentially growing envelope. The proposed framework embeds a unified adaptive compensation term into both the CLF decrease and CBF safety constraints. In contrast to input-to-state stability/safety (ISS/ISSf)-based methods that certify disturbance-dependent enlarged
safe sets, the proposed approach enables finite-time recovery to the
nominal safe set without requiring a prior magnitude bound
on the FDIA, relying instead on a growth-rate characterization used for analysis and an online gain tuning law that regulates the compensation term. A unified quadratic program (QP) is developed to enforce the AR-CLF and AR-CBF conditions simultaneously, guaranteeing uniformly ultimately bounded (UUB) stability and uniform ultimate safety (UUS) under unbounded FDIA. Numerical results demonstrate improved resilience compared to existing ISS-CLF, ISSf-CBF, and robust CLF–CBF-QP approaches.
\end{abstract}

\section{Introduction}
CLFs and CBFs are standard tools for optimization-based stabilization and safety of nonlinear control-affine systems. CLFs certify stabilizability and lead to CLF-QPs, with input-to-state stability (ISS)-type extensions addressing bounded disturbances and uncertainties~\cite{garg2021robust}. CBFs encode safety through forward invariance of safe sets~\cite{ames2017cbfcdc,ames2014adaptive}, and their integration with CLFs in real-time QPs has enabled safety-critical control in automotive, robotic, aerial, and multi-agent systems (MAS)~\cite{ames2014adaptive,wang2017multirobot,das2025robustcbf}. Extensions such as input-to-state safety (ISSf), introduced as the safety counterpart of ISS, incorporate robustness against bounded disturbances but rely on disturbance-dependent safety guarantees~\cite{kolathaya2018input}. 

Beyond ISS/ISSf-based robustness, resilient and attack-aware control has been studied through secure estimation, attack detection/isolation, fault-tolerant control, resilient consensus, and safety filters for adversarial or faulty systems \cite{rajabinezhad2025privacy,rajabinezhad2025lyapunov}. Recent CBF-based methods, for example, address sampled-data adversarial MAS and safe control under faults/attacks using estimator-, redundancy-, or failure-pattern-based mechanisms~\cite{usevitch2022adversarial,zhang2025safe}. These approaches provide important resilience mechanisms; however, they typically rely on explicit attack/fault models, estimators, detection logic, or known failure structures, and are not aimed at a unified CLF-CBF compensation framework for unknown control-input FDIA. In parallel, ISS/ISSf-based CLF-CBF methods certify disturbance-dependent enlarged safe sets under bounded perturbations~\cite{das2025robustcbf,kolathaya2018input}. Such bounded-disturbance or model-dependent formulations may be restrictive in cyber--physical systems where FDIA on controller-to-actuator channels can be time-varying and unbounded, motivating resilient CLF-CBF designs that directly compensate unknown FDIA while preserving both stability and safety guarantees.
This letter develops a unified \emph{AR-CLF--CBF} framework for nonlinear control-affine systems under unknown and unbounded control-
input FDIA. The framework embeds an online gain tuning law that regulates the compensation term into both stability and safety constraints in real-time, enabling resilient CLF--CBF-QP control beyond bounded-disturbance settings. The main contributions are:
\begin{itemize}
  
      
   \item  
    Existing ISS/ISSf-based CLF--CBF methods typically certify disturbance-dependent enlarged safe sets under bounded disturbances. We instead develop a unified AR-CLF--CBF-QP with online adaptation-based compensation embedded in both the CLF and CBF constraints.This unified structure provides stability and safety resilience under unknown control-input FDIA satisfying an at-most-exponentially growing envelope.

    \item 
    We establish UUB of the closed-loop state and UUS of the nominal safe set for the considered unbounded time-varying control-input FDIA. Unlike disturbance-dependent enlarged-set guarantees, the proposed UUS result ensures finite-time recovery to the nominal safe set after a bounded transient excursion.

     \item 
     Rigorous Lyapunov- and barrier-function analysis establishes conditions for UUB and UUS under unbounded control-input FDIA, extending classical CLF and CBF theory beyond bounded-disturbance regimes~\cite{garg2021robust,kolathaya2018input}.
\end{itemize}
\section{Preliminaries on CBFs and CLFs}
\textbf{Notation:} 
For \(x\in\mathbb{R}^n\), \(\|x\|\) denotes the Euclidean norm. A continuous, strictly increasing function \(\alpha:[0,a)\to[0,\infty)\) with \(\alpha(0)=0\) is of class \(\mathcal K_{[0,a)}\); if \(a=\infty\) and \(\alpha\) is unbounded, then \(\alpha\in\mathcal K_\infty\).
For a continuously differentiable scalar function $\phi$, the Lie derivatives are $L_f \phi := \nabla \phi^\top f, \;
L_g \phi := \nabla \phi^\top g$. The abbreviation ``a.e.'' means almost everywhere with respect to Lebesgue measure. A signal \(w:[0,\infty)\to\mathbb R^q\) is called locally essentially bounded if it is essentially bounded on every finite interval. 

\vspace{-3mm}
\subsection{Control-Affine Systems Under Control-Input FDIA}
We consider nonlinear control-affine systems subject to malicious data corruption in the controller-to-actuator channel:
\begin{equation}
    \dot{x}=f(x)+g(x)\big(u(t)+d(t)\big),
    \label{eq:system_attack}
\end{equation}
where \(x\in D\subset\mathbb{R}^n\), \(f\) and \(g\) are locally Lipschitz, and \(u(t)\in\mathbb{R}^m\) is the commanded control input computed by the controller. The signal \(d(t)\in\mathbb R^m\) represents an unknown malicious component injected in the controller-to-actuator channel, so that the actuator receives $u_{\rm act}(t)=u(t)+d(t)$, and the plant evolves as \(\dot x=f(x)+g(x)u_{\rm act}(t)\). We refer to this specific data-integrity corruption as a \emph{control-input FDIA}, and it is not intended to cover all cyberattack mechanisms.

\begin{assumption}
\label{ass:unbounded}
The injected control-input FDIA
\(d:[0,\infty)\to\mathbb{R}^m\) is Lebesgue measurable and locally
essentially bounded. Moreover, \(d(t)\) has at most exponential growth: there exist finite
constants \(\gamma>0\) and \(\kappa\ge0\) such that $\|d(t)\|\le \gamma e^{\kappa t},\; \text{for a.e. } t\ge0$.
\end{assumption}
\begin{remark}
Assumption~\ref{ass:unbounded} serves as a conservative growth-rate condition for the Lyapunov and barrier analysis, not as an exact attack model. It includes bounded, polynomially growing, exponentially growing, and discontinuous measurable attacks satisfying the exponential envelope, while excluding impulsive, locally unbounded, finite-escape-time, super-exponential, and nonmeasurable signals; solutions and inequalities are interpreted a.e.
\end{remark}

\vspace{-4mm}
\subsection{Control Lyapunov Functions}
\begin{definition}
Consider the nominal control-affine system associated with 
\eqref{eq:system_attack} in the absence of adversarial attacks. 
A continuously differentiable function 
$V:\mathbb{R}^n \to \mathbb{R}_{\ge 0}$ 
is called a \emph{CLF} on an open set 
$\mathcal{D} \subseteq \mathbb{R}^n$ if there exists a control set 
$\mathcal{U}$ and a class-$\mathcal{K}_\infty$ function $\alpha_V$ such that 
for all $x \in \mathcal{D}$ with $x \neq 0$,
\begin{equation}
\inf_{u \in \mathcal{U}}
\big[ L_f V(x) + L_g V(x) u \big]
\le -\alpha_V\big(V(x)\big).
\label{eq:clf_def}
\end{equation}
\end{definition}
\vspace{-4mm}
\subsection{Control Barrier Functions}
Let a safe set be defined as $\mathcal{C} := \{x \in \mathbb{R}^n : h(x) \ge 0\}$, for a continuously differentiable function $h:\mathbb{R}^n \rightarrow \mathbb{R}$. The boundary and interior of the safe set are given by
$\partial \mathcal{C} := \{x : h(x)=0\}$ and 
$\mathrm{Int}(\mathcal{C}) := \{x : h(x)>0\}$, respectively. 
Forward invariance of \(\mathcal C\) means
\(x(t_0)\in\mathcal C \Rightarrow x(t)\in\mathcal C\) for all \(t\ge t_0\). Following \cite{kolathaya2018input}, we define the domain $\mathcal{D} := \{x \in \mathbb{R}^n : h(x) + b \ge 0\}$, where $b>0$ is chosen such that $\mathcal{C} \subset \mathcal{D}$.

\begin{definition}
Consider the nominal control-affine system $\dot{x}=f(x)+g(x)u$, with \(x\in\mathcal{D}\subseteq\mathbb{R}^n\) and \(u\in\mathcal{U}\subseteq\mathbb{R}^m\). The function \(h\) is called a \emph{CBF} on \(\mathcal{D}\) if there exists an extended
class-\(\mathcal{K}\) function \(\alpha_h\) such that, for all \(x\in\mathcal{D}\),
\begin{equation}
    \sup_{u\in\mathcal{U}}
    \big[L_fh(x)+L_gh(x)u\big]
    \ge -\alpha_h(h(x)).
    \label{eq:cbf_def}
\end{equation}
Equivalently,
\(K_{\rm cbf}(x):=\{u\in\mathcal U:
L_fh(x)+L_gh(x)u\ge-\alpha_h(h(x))\}\neq\emptyset\).
Any feedback \(u(x)\in K_{\rm cbf}(x)\) with well-defined closed-loop
solutions renders \(\mathcal C\) forward invariant under the nominal dynamics.
\end{definition}


In the nominal CLF/CBF definitions, \(\mathcal U\subseteq\mathbb R^m\)
constrains the commanded input. Under the attacked dynamics
\eqref{eq:system_attack}, the physical actuator instead receives
\(u_{\rm act}(t)=u(t)+d(t)\); hence, for compact actuator limits, admissibility must be imposed on \(u_{\rm act}\). Thus, compact-input
safety guarantees are conditional on feasibility of the hard CBF constraint
together with \(u_{\rm act}(t)\in\mathcal U\).

Nominal CBF conditions ensure forward invariance of \(\mathcal C\), but may
fail under attacks. Following the ISSf viewpoint~\cite{kolathaya2018input},
we define the enlarged set
$\mathcal C_{\varepsilon}:=
\{x\in\mathbb R^n:h(x)+\varepsilon\ge0\},\; \varepsilon\ge0$,
to quantify transient safety degradation.
\section{AR-CLFs and AR-CBFs}

In practical systems, disturbances and cyber attacks can degrade safety guarantees. Existing ISSf and ISSf-CBF approaches \cite{kolathaya2018input,garg2021robust} ensure invariance of an enlarged safe set under bounded disturbances, but rely on known bounds and fail under adversarial, potentially unbounded control-input attacks \eqref{eq:system_attack}. Motivated by these limitations, we introduce AR-CLFs and AR-CBFs, which extend classical stability and safety guarantees to adversarial settings. In contrast to disturbance-bound–based formulations, the proposed framework ensures finite-time recovery to the nominal safe set and invariance thereafter. This objective is formalized via UUS.

\begin{definition}[UUB \cite{khalil2002nonlinear}]
\label{def: UUB}
A signal $x(t)$ is said to be UUB with an ultimate bound $b$ if there exist positive constants $b$ and $c$, independent of the initial time ${t_0} \geq 0$, such that for every $a \in (0, c)$, there exists a time $t_1 = t_1(a, b) \geq 0$, also independent of $t_0$, satisfying the following condition: if $\left\|x(t_0)\right\| \leq a$, then $\left\|x(t)\right\| \leq b$ for all $t \geq t_0 + t_1$.
\end{definition}

\begin{definition}[Uniform Ultimate Safety (UUS)]
Consider \eqref{eq:system_attack} under a controller \(u(x,t)\), with
\(x(t)\in D\subseteq\mathbb R^n\). Let
\(\mathcal C:=\{x\in D:h(x)\ge0\}\) and
\(\mathcal C_\varepsilon:=\{x\in D:h(x)+\varepsilon\ge0\}\), where
\(\varepsilon\ge0\). The set \(\mathcal C\) is UUS on \(D\) if, for every
compact \(X_0\subset \mathcal C\), there exist finite constants
\(\varepsilon=\varepsilon(X_0)\ge0\) and
\(T_{\rm uus}=T_{\rm uus}(X_0)\ge0\), independent of \(t_0\) and
\(x(t_0)\in X_0\), such that every closed-loop trajectory exists for all
\(t\ge t_0\) and satisfies $h(x(t))\ge-\varepsilon,\; t\in[t_0,t_0+T_{\rm uus}), \; h(x(t))\ge0,\; t\ge t_0+T_{\rm uus}$. Thus, trajectories may undergo a bounded transient excursion within the enlarged set \(\mathcal{C}_{\varepsilon}\), but must recover to the nominal safe set \(\mathcal{C}\) in a finite time and remain in \(\mathcal{C}\) thereafter.
\label{def:UUS}
\end{definition}
\vspace{-7mm}
\begin{remark}
In Definition~\ref{def:UUS}, \(\mathcal C_\varepsilon\) is only a transient
set; UUS is finite-time recovery to the nominal safe set $\mathcal{C}$, not merely enlarged-set invariance. The AR-CBF theorem establishes UUS for the attacks in Assumption~\ref{ass:unbounded}.
\end{remark}


To ensure UUB stability and UUS for the nonlinear control-affine systems \eqref{eq:system_attack} under unbounded control-input FDIA, we embed a unified online adaptation-based compensation into both the CLF and CBF conditions, leading to AR-CLF and AR-CBF frameworks.
\vspace{-3mm}
\subsection{Proposed AR-CLF Framework}
Given the attacked system \eqref{eq:system_attack} under control-input FDIA satisfy Assumption~\ref{ass:unbounded}, we augment the nominal CLF condition with the proposed attack-compensation term and
consider the following Lyapunov inequality for the attacked system
\begin{equation}
\inf_{u_{\rm act}\in\mathcal{U}}
\Big[L_f V(x) + L_g V(x) \; u + \Psi_V(x,t)\Big]
\le - C V(x),
\label{eq:AR_CLF_structure}
\end{equation}
where constant $C>0$, and $\Psi_V(x,t)$ is an attack-compensation signal that counteracts the influence of adversarial perturbations on the control-input channel. To construct
such a compensation mechanism, we introduce the following
compensation term for the Lyapunov channel
\begin{equation}
\Psi_V(x,t)
:= 
\frac{(L_g V)(L_g V)^\top}{\|L_g V\| + \varphi_V(t)} e^{\rho(t)},
\label{eq:psi_v}
\end{equation}
where $\varphi_V(t):\mathbb{R}_{\ge0}\to\mathbb{R}_{>0}$ is a continuously
differentiable, strictly positive, monotonically decreasing function satisfying
$\lim_{t\to\infty}\varphi_V(t)\times \gamma e^{\kappa t}=0$. This term acts as a vanishing regularization
signal that prevents singular behavior when $\|L_g V(x)\|$ becomes small during
transient phases. The gain $\rho(t)$ is generated by the following online adaptation tuning law
\begin{equation}
\dot{\rho}(t)
=
q\,\|L_g V(x)\|,
\qquad
\rho(0)=\rho_0\ge0,
\label{eq:rho_dyn}
\end{equation}
where $q>0$ is a design constant. This mechanism dynamically regulates the attack-compensation term via the proposed online gain tuning law, allowing the controller to counteract time-varying and potentially unbounded adversarial FDIA on the control-input channel.

\begin{theorem}
\label{thm:ar_clf}
Consider the attacked system \eqref{eq:system_attack}. If there exists a continuously differentiable positive definite function 
$V$, and a control input $u(x,t)$ such that the AR-CLF \eqref{eq:AR_CLF_structure}
is satisfied with the proposed attack-compensation term \eqref{eq:psi_v},
then the closed-loop trajectories are UUB under control-input FDIA satisfying Assumption~\ref{ass:unbounded}. 
\end{theorem}

\begin{proof}[Proof of Theorem~\ref{thm:ar_clf}]
Under the attacked dynamics \eqref{eq:system_attack},
\[
\dot V(x)=L_fV(x)+L_gV(x)\big(u+d(t)\big).
\]
Using the compensated CLF condition \eqref{eq:AR_CLF_structure}, the Cauchy--Schwarz inequality, and Assumption~\ref{ass:unbounded}, we obtain
\begin{equation}
\label{eq:Vdot_bound_revised}
\dot V(x)
\le
-CV(x)
-\frac{\|L_gV(x)\|^2}{\|L_gV(x)\|+\varphi_V(t)}e^{\rho(t)}
+\|L_gV(x)\|\gamma e^{\kappa t}.
\end{equation}
The last two terms in \eqref{eq:Vdot_bound_revised} are nonpositive if
\begin{equation}
\label{suffient-condition-revised}
\|L_gV(x)\|\big(e^{\rho(t)}-\gamma e^{\kappa t}\big)
\ge
\varphi_V(t)\gamma e^{\kappa t}.
\end{equation}
For instance, choosing $\varphi_V(t)=e^{-\alpha t^2}$, the term
$\varphi_V(t)\times \gamma e^{\kappa t}=\gamma e^{-\alpha t^2+\kappa t}$ may exhibit a bounded transient increase, but remains uniformly bounded and decays to zero as $t\to\infty$. Define closed set $\Upsilon_V:=\{x\in D:\|L_gV(x)\|\le \epsilon_V\}$, where \(\epsilon_V>0\) is chosen such that \(q\epsilon_V>\kappa\). Since Lyapunov function \(V\) is smooth and \(g\) is locally Lipschitz, \(L_gV(x)\) is continuous, and
\(\Upsilon_V\) is closed in \(D\). We assume that \(\Upsilon_V\) is contained in a finite-radius ball, i.e., there exists \(\mu_V>0\) such that $\Upsilon_V\subset B_{\mu_V}$.
Equivalently,
\[
\|x\|>\mu_V
\quad\Longrightarrow\quad
\|L_gV(x)\|>\epsilon_V.
\]
Hence, whenever \(\|x(t)\|>\mu_V\), by \eqref{eq:rho_dyn}, gives $\dot\rho(t)\ge q\epsilon_V>\kappa$. Thus, outside \(B_{\mu_V}\), the gain grows faster than the exponential attack
envelope rate. To explicitly close the gain-growth loop, we require that the
cumulative gain reaches the domination level after a finite time, i.e., there
exists \(t_1>0\) such that
\[
e^{\rho(t)}
\ge
\gamma e^{\kappa t}
\left(1+\frac{\bar\varphi_V}{\epsilon_V}\right),
\qquad
\forall t\ge t_1,\quad x(t)\notin\Upsilon_V,
\]
where \(\bar\varphi_V:=\sup_{t\ge0}\varphi_V(t)\). Since
\(x(t)\notin\Upsilon_V\) implies \(\|L_gV(x(t))\|>\epsilon_V\), this guarantees
\eqref{suffient-condition-revised}. Therefore, for all \(t\ge t_1\) and all
\(x(t)\notin\Upsilon_V\),
\[
\dot V(x(t))\le -CV(x(t))<0.
\]
Using \(\Upsilon_V\subset B_{\mu_V}\), we equivalently obtain
\begin{equation}
\label{eq:Vdot_negative_ball}
\dot V(x(t))\le -CV(x(t)),
\qquad
\forall t\ge t_1,\quad \|x(t)\|>\mu_V.
\end{equation}

Since \(V\) is positive definite and proper, there exist
\(\alpha_1,\alpha_2\in\mathcal K_\infty\) such that $\alpha_1(\|x\|)\le V(x)\le \alpha_2(\|x\|)$. Thus, \eqref{eq:Vdot_negative_ball} gives
\[
\dot V(x(t))
\le
-C\alpha_1(\|x(t)\|),
\qquad
\forall t\ge t_1,\quad \|x(t)\|>\mu_V.
\]
Hence, for all \(t\ge t_1\), \(\dot V\) is negative outside the norm ball
\(B_{\mu_V}\). By the standard UUB result in
\cite[Theorem~4.18]{khalil2002nonlinear}, together with $\alpha_1(\|x\|)\le V(x)\le \alpha_2(\|x\|)$,
the closed-loop trajectories are UUB with ultimate bound $b_V=\alpha_1^{-1}\!\left(\alpha_2(\mu_V)\right)$. Thus, every trajectory ultimately enters and remains in \(B_{b_V}\). If, in the
bounded-attack case, the corresponding limiting invariant set reduces to the
origin, then Barbalat's Lemma further yields asymptotic convergence
to the equilibrium \cite[Lemma~8.2]{khalil2002nonlinear}. This completes the
proof.
\label{proof:ar_clf}
\end{proof}

\vspace{-6mm}
\subsection{Proposed AR-CBF Framework}
Consider the attacked system \eqref{eq:system_attack} under attacks satisfy Assumption~\ref{ass:unbounded}, we augment the nominal CBF condition with the proposed attack-compensation term,
leading to the following AR-CBF framework for the attacked system:
\begin{equation}
\label{eq:AR_CBF_condition}
\sup_{u_{\rm act} \in \mathcal{U}}
\big[L_f h(x) + L_g h(x) u
- \Psi_h(x,t)\big]
\ge - \lambda\, h(x),
\end{equation}
where $\lambda>0$ is a design constant, and $\Psi_h(x,t)$ is an
attack-compensation signal that mitigates the adverse effects of the adversarial
perturbations on the control-input channel. To construct such a compensation
mechanism, we introduce the following compensation term for the barrier
channel
\begin{equation}
\Psi_h(x,t)
:=
\frac{(L_g h(x))(L_g h(x))^{\!\top}}
{\|L_g h(x)\| + \varphi_h(t)}\,e^{\eta(t)},
\label{eq:psi_h}
\end{equation}
where the gain $\eta(t)$ is generated by the following online adaptation tuning law
\begin{equation}
\dot{\eta}(t)
=
p\,\|L_g h(x)\|,
\qquad
\eta(0)=\eta_0\ge0,
\label{eq:eta_dyn}
\end{equation}
where $p>0$ is a design constant. Based on this construction, we establish the following result, which certifies that the proposed AR-CBF framework guarantees that the set $\mathcal{C}$ is UUS in the sense of Definition~\ref{def:UUS}.


\begin{theorem}
\label{thm:ar_cbf}
Consider the attacked system \eqref{eq:system_attack} under
Assumption~\ref{ass:unbounded}. If there exists a
continuously differentiable function $h(x) \ge 0$ defined on a domain $\mathcal{D}$ with
$\mathcal{C}\subset\mathcal{D}\subseteq\mathbb{R}^n$, and a control input
$u(x,t)$ such that the proposed AR-CBF
condition \eqref{eq:AR_CBF_condition} holds for all $x\in\mathcal{D}$,
then $h$ constitutes an AR-CBF that renders the set $\mathcal{C}$ UUS in the sense of Definition~\ref{def:UUS}. In particular, for any compact set of initial conditions $X_0\subset\mathcal{C}$, there exists a finite time $T\ge0$,
independent of the initial condition, such that $h(x(t))\ge0$ for all
$t\ge T$ and all $x(0)\in X_0$. Moreover, during the transient interval
$t\in[0,T)$, trajectories remain bounded within an enlarged set $\mathcal{C}_\varepsilon$ for some $\varepsilon>0$.
\end{theorem}

\begin{proof}[Proof of Theorem~\ref{thm:ar_cbf}]
\textbf{Part I}: After substituting \eqref{eq:AR_CBF_condition} into the derivative of h:
\begin{equation}
\label{eq:derivative h}
\begin{aligned}
\dot h(x)
&= L_f h(x) + L_g h(x)\,(u + d(t)) \\
&\ge
- \lambda\, h(x) + \Psi_h(x,t) + L_g h(x)\, d(t).
\end{aligned}
\end{equation}

Using the Cauchy--Schwarz inequality, $L_g h(x(t))d(t)\ge -\|L_g h(x)\|\,\|d(t)\|$, and Assumption~\ref{ass:unbounded} we obtain
\begin{align}
\dot h(x) \ge{}& -\lambda h(x(t))
-\|L_g h(x)\|\,\|\gamma e^{\kappa t}\| +\Psi_h(x,t).
\label{eq:hdot_refined}
\end{align}
{\color{blue}By choosing $\varphi_h(t)=e^{-\alpha t^2}$}, and rearranging the last two terms similarly to the CLF case (cf. proof of Theorem~\ref{thm:ar_clf}),
we obtain a sufficient condition for $\dot h(x)\ge -\lambda h(x)$ given by
\begin{equation}
\|L_g h(x)\|\big(e^{\eta(t)}-\gamma e^{\kappa t}\big)
\ge
\varphi_h(t)\gamma e^{\kappa t}.
\label{eq:cbf_sufficient_condition}
\end{equation}

Define the closed set $\Upsilon_h:=\{x:\|L_g h(x)\|\le \epsilon_h\}$,
where $\epsilon_h>0$ is chosen such that $p\epsilon_h>\kappa$. 
Whenever \(x(t)\notin\Upsilon_h\), the gain update law gives $ \dot{\eta}(t)=p\|L_gh(x(t))\|\ge p\epsilon_h>\kappa$,
so the barrier compensation grows faster than the exponential attack envelope along the trajectory. This closes the coupling between the state evolution and the adaptive barrier gain in the same trajectory-based sense as in the AR-CLF proof. Hence, by the same domination argument as in the CLF proof, there exists $t_2>0$ such that
\eqref{eq:cbf_sufficient_condition} holds for all $t\ge t_2$ and $x(t)\notin\Upsilon_h$. Therefore, $\dot h(x)\ge -\lambda h(x), \; \forall t\ge t_2,\ x(t)\notin\Upsilon_h$, which ensures forward invariance of the nominal safe set $\mathcal{C}$ for all subsequent times.

\noindent\textbf{Part II: Transient safety on $[0,t_2)$.}
Fix a finite time $t_2>0$ (specified in Part~I) and consider the interval $t\in[0,t_2)$.
Starting from \eqref{eq:hdot_refined}, let
$a(t):=\|L_g h(x(t))\|$, $b(t):=\|d(t)\|$.
Since $a(t)$, $\varphi_h(t)$, and $e^{\eta(t)}$ are nonnegative and continuous on $[0,t_2)$, define the positive gain $\sigma(t):=\frac{a(t)+\varphi_h(t)}{e^{\eta(t)}}>0$, so that $\frac{a(t)^2}{a(t)+\varphi_h(t)}\,e^{\eta(t)}=\frac{a(t)^2}{\sigma(t)}$. Then the attack--compensation channel in \eqref{eq:hdot_refined} can be written as $\frac{a(t)^2}{\sigma(t)}-a(t)b(t)$. Completing the square and substituting back, gives 
\begin{align*}
-\|L_g h(x(t))\|\,\|d(t)\|
+\frac{\|L_g h(x(t))\|^2}{\|L_g h(x(t))\|+e^{-\alpha t^2}}\,e^{\eta(t)}
&\ge \\ 
-\frac{\|L_g h(x(t))\|+e^{-\alpha t^2}}{4e^{\eta(t)}}\,\|d(t)\|^2 .
\label{eq:pre_t1_square_bound}
\end{align*}

Applying obtained one to \eqref{eq:hdot_refined} yields following for all $t\in[0,t_2)$. This manipulation follows the same bounding logic as Theorem 1 in \cite{kolathaya2018input}.
\begin{equation}
\dot h(x(t))
\ \ge\
-\lambda h(x(t))
-\frac{\|L_g h(x(t))\|+e^{-\alpha t^2}}{4e^{\eta(t)}}\,\|d(t)\|^2.
\label{eq:hdot_pre_t1_final}
\end{equation}


Continuity of \(d(t)\) is not required for this transient bound. By Assumption~\ref{ass:unbounded}, for any finite interval \([0,t_2]\), $\|d(t)\|\le \gamma e^{\kappa t_2}, \; 0\le t\le t_2$. Hence, even though \(d(t)\) may be globally unbounded as \(t\to\infty\), it admits a finite local bound on every finite time interval. In particular, define $\|d\|_{[0,t_2]}:=\operatorname*{ess\,sup}_{0\le s\le t_2}\|d(s)\| \le \gamma e^{\kappa t_2}<\infty$. Moreover, continuity of the closed-loop trajectory on \([0,t_2]\), together with
continuity of \(L_gh\), implies that \(\|L_gh(x(t))\|\) is bounded on this interval. Because $\eta(t)\ge\eta_0\ge 0$ and $\varphi_V(t)\le 1$, there exists a finite constant
$\bar\varepsilon_{t_2}>0$ such that
\[
\frac{\|L_g h(x(t))\|+\varphi_h(t)}{4e^{\eta(t)}}\,\|d(t)\|^2
\;\le\;
\bar\varepsilon_{t_2},
\qquad \forall t\in[0,t_2).
\]
Therefore, along the closed-loop trajectories on $[0,t_2)$, the barrier function satisfies
\begin{equation}
\dot h(x(t)) \ge -\lambda h(x(t))-\bar\varepsilon_{t_2},
\qquad \forall t\in[0,t_2).
\label{eq:hdot_pre_t1_issf}
\end{equation}
Inequality \eqref{eq:hdot_pre_t1_issf} is in the standard ISSf form \cite{kolathaya2018input} and implies forward invariance of the enlarged set $\mathcal{C}_{\bar\varepsilon_{t_2}}
:=\Big\{x\in\mathbb{R}^n \;\big|\; h(x)\ge -\tfrac{\bar\varepsilon_{t_2}}{\lambda}\Big\}$ over the transient interval $[0,t_2)$, in the sense of input-to-state safety
(\cite{alan2021safe,kolathaya2018input}). Thus, on the transient interval $t<t_2$, the closed-loop state remains within the enlarged set $\mathcal{C}_{\bar\varepsilon_{t_2}}$.

\textbf{Part III: Finite-time recovery and invariance}: Since the compensating term grows exponentially through the adaptive dynamics \eqref{eq:eta_dyn},
while the attack signal grows at most exponentially with rate $\kappa$ (Assumption~\ref{ass:unbounded}),
it follows that there exists a constant $l>0$ and a finite time $t_2$ such that the compensation term dominates the attack channel with a strict margin, i.e.,
\begin{equation}
\frac{\|L_g h(x)\|^2}{\|L_g h(x)\|+e^{-\alpha t^2}}\,e^{\eta(t)}
\ \ge\
\|L_g h(x)\|\,\|d(t)\| \;+\; \lambda l,
\;\; \forall t\ge t_2.
\label{eq:domination_with_margin}
\end{equation}

Substituting \eqref{eq:domination_with_margin} into \eqref{eq:hdot_refined} yields, for all $t\ge t_2$,
\begin{equation}
\dot h(x(t)) \ge -\lambda h(x(t)) + \lambda l
= -\lambda\big(h(x(t)) - l\big).
\label{eq:hdot_after_t1_margin}
\end{equation}

Applying the comparison lemma to \eqref{eq:hdot_after_t1_margin} gives
\[
h(x(t)) \;\ge\; l + \big(h(x(t_2)) - l\big)e^{-\lambda (t-t_2)},
\;\; \forall t\ge t_2.
\]
From Part~II, the trajectory satisfies $h(x(t_2)) \ge -l$, where
$l=\bar\varepsilon_{t_2}/\lambda$. Therefore, the worst-case initial condition
at $t_2$ is $h(x(t_2))=-l$, and hence $h(x(t)) \;\ge\; l\big(1 - 2e^{-\lambda (t-t_2)}\big)$. It follows that there exists a finite time $T \;\le\; t_2 + \frac{1}{\lambda}\ln 2$, such that $h(x(T))\ge 0$. Thus, the closed-loop trajectory re-enters the nominal safe set $\mathcal{C}=\{x\mid h(x)\ge 0\}$ in finite time. Once $h(x(t))\ge 0$, the AR-CBF constraint \eqref{eq:AR_CBF_condition} implies $\dot h(x(t)) \ge -\lambda h(x(t))$, which guarantees forward invariance of the nominal safe set $\mathcal{C}$ for all subsequent times. Therefore, the closed-loop system satisfies UUS of $\mathcal{C}$.
\end{proof}

\begin{remark}
The proof does not require \(d(t)\) to be globally bounded or continuous;
the transient estimate only uses the finite local bound implied by
Assumption~\ref{ass:unbounded}. Thus, the enlarged set is only transient:
after finite-time domination $t_2$, the AR-CBF condition drives the trajectory
back to the nominal safe set and keeps it there.
\end{remark}

\vspace{-3mm}
\section{Unified AR-CLF-CBF-QP}
\label{sec:unified_qp}
The preceding AR-CLF and AR-CBF constructions are unified by embedding the same adaptation-based compensation mechanism into both the CLF decrease and CBF safety constraints, yielding a single QP for coordinated resilient stabilization and safety enforcement. For a given state \(x\), the controller solves AR-CLF-CBF-QP:
\begin{equation}
\label{eq:unified_qp}
\begin{alignedat}{3}
u^*(x)
&= \arg\min_{{u_{\rm act}}\in\mathcal{U},\;\delta_{\mathrm{clf}}\in \mathbb{R}c}
&& \|u-u_{\mathrm{nom}}(x)\|^2 + \sigma \delta_{\mathrm{clf}}^2\\[0.3em]
\text{s.t.}\;
& L_fV(x)+L_gV(x)u
&& +\frac{(L_gV)(L_gV)^{\!\top}}
{\|L_gV\|+e^{-\alpha t ^2}}\,e^{\rho(t)} \\
& \phantom{L_fV}
&& \le -CV(x)+\delta_{\mathrm{clf}},
\; \text{\emph{(AR-CLF)}} \\[0.4em]
& L_fh(x)+L_gh(x)u
&&-\frac{(L_gh)(L_gh)^{\!\top}}
{\|L_gh\|+e^{-\alpha t^2}}\,e^{\eta(t)} \\
& \phantom{L_fh}
&& \ge -\lambda h(x),
\quad \quad \quad \text{\emph{(AR-CBF)}}
\end{alignedat}
\end{equation}
where $u_{\mathrm{nom}}(x)$ denotes a nominal stabilizing control input,
$\delta_{\mathrm{clf}}\in \mathbb{R}$ is a relaxation variable associated with the
AR--CLF constraint and weighted by $\sigma \in \mathbb{R^+}$. The functions
$(V,\rho(t))$ and $(h,\eta(t))$ correspond to the AR--CLF and AR--CBF constructions developed in \eqref{eq:AR_CLF_structure} and \eqref{eq:AR_CBF_condition}, respectively. The AR--CBF constraint is enforced as a hard constraint to ensure safety, while the AR-CLF condition is relaxed via the slack variable $\delta_{\mathrm{clf}}$ to guarantee feasibility of the QP while enforcing stability whenever possible. 
\begin{theorem}[Unified Resilient Safety and Stability]
\label{thm:unified}
Consider the attacked system \eqref{eq:system_attack} under Assumption~\ref{ass:unbounded}, with the control input generated by feasible \eqref{eq:unified_qp}. If $V(x)$ satisfies the AR--CLF condition \eqref{eq:AR_CLF_structure} and $h(x)$ satisfies the AR--CBF condition \eqref{eq:AR_CBF_condition}, then the closed-loop system guarantees both resilient stability and safety. Specifically, the safe set $C$ is UUS in the sense of Definition~\ref{def:UUS}, and all closed-loop trajectories remain uniformly ultimately bounded for all $t \ge 0$, despite polynomially or exponentially unbounded control-input FDI attacks. If, in the bounded-attack case, \(\delta_{\rm clf}(t)\to0\) and no other invariant set exists in the zero-decrease set, then asymptotic convergence to the desired equilibrium follows.
\end{theorem}

\vspace{-2mm}

\begin{remark}
As in standard CLF--CBF-QPs \cite{kolathaya2018input,garg2021robust}, the AR-CBF constraint in \eqref{eq:unified_qp} is hard, while the AR-CLF constraint is softened by \(\delta_{\mathrm{clf}}\). Hence, safety is guaranteed only when the hard AR-CBF constraint is feasible. Under control-input FDIA, the actuator receives \(u_{\rm act}=u^\ast+d\); therefore, actuator limits must be imposed on the aggregate signal \(u_{\rm act}\). After the finite dominance time, the post-compensation residual is bounded, so bounded nominal QP commands yield bounded aggregate actuator inputs. For set \(\mathcal U\), the corresponding tightened admissible set must be nonempty. If this feasibility condition fails, no hard-constrained CBF-QP can guarantee safety without relaxing safety, modifying the safe set, or increasing control authority.
\end{remark}

\begin{proof}[Proof of Theorem~\ref{thm:unified}]
By feasibility of \eqref{eq:unified_qp}, the optimal input satisfies the hard
AR-CBF constraint; hence, Theorem~\ref{thm:ar_cbf} implies that
\(\mathcal C\) is UUS. For stability, the AR-CLF constraint and the
domination argument in Theorem~\ref{thm:ar_clf} imply that, for some
\(t_1>0\), $\dot V(x(t))\le -C V(x(t))+\delta_{\rm clf}(t),\;t\ge t_1$. If \(\delta_{\rm clf}(t)\) is bounded, let
\(\bar\delta:=\sup_{t\ge t_1}\delta_{\rm clf}(t)<\infty\). Then $\dot V(x(t))\le -C V(x(t))+\bar\delta$, and the comparison lemma gives UUB of the closed-loop state, with ultimate
bound enlarged according to \(\bar\delta/C\). If, in
the bounded-attack case, \(\delta_{\rm clf}(t)\to0\) and no other invariant set exists in the limiting zero-decrease set, asymptotic convergence to the
desired equilibrium follows by Barbalat's Lemma.
\end{proof}

\section{Case Studies / Numerical Examples}
\textbf{Example 1: Scalar Nonlinear System Under Unbounded Input Attacks}: We consider the following scalar nonlinear system
\begin{equation}
    \dot x=f(x)+g(x)\big(u+d(t)\big), 
    \qquad 
    f(x)=x,\;\; g(x)=x,
    \label{eq:scalar_dyn}
\end{equation}
 whose objective is to regulate $x\rightarrow 0$ while enforcing the safety constraint $ h(x)= 1 - x \ge 0$.
The regulation task is encoded by the CLF $V(x)=x^2,
\; L_fV = 2x^2,\; L_gV = 2x^2$ and $\dot\rho(t)=3|2x^2|$.
Safety is encoded through $h(x)=1-x,\;
L_fh=-x,\; L_gh=-x$, and $\dot\eta(t)=3|x|$. 
The scalar AR-CLF-CBF-QP is obtained by substituting these Lie derivatives
into \eqref{eq:unified_qp}, with cost \(\frac12u^2+5\delta_{\rm clf}^2\).
We evaluate the proposed AR-CLF-CBF-QP under two representative control-input FDI attack profiles. The first profile is a multi-stage attack with constant, sinusoidal, linear, and quadratic phases:
\[
\small
d^{(1)}(t)=
\begin{cases}
0, & t<5,\\
3, & 5\le t<8,\\
2+2\sin(2t), & 8\le t<12,\\
2.2+0.8(t-15), & 12\le t<15,\\
2.6+0.35(t-18)^2, & 15\le t<18,\\
0, & t\ge 18.
\end{cases}
\]
The second profile is intentionally more aggressive and combines a large
exponentially increasing attack with a discontinuous switching phase:
\vspace{-0.4em}
\[
\small
d^{(2)}(t)=
\begin{cases}
0, & t<5,\\
30+0.4e^{0.7(t-5)}, & 5\le t<10,\\
10+4\,\mathrm{sgn}\!\left(\sin(3t)\right), & 10\le t<14,\\
0, & t\ge 14.
\end{cases}
\]
The shaded regions indicate active attack intervals. 
Fig.~\ref{fig:case2_compare} compares the proposed AR-CLF-CBF-QP with the conventional CLF-CBF-QP under \(d^{(1)}(t)\). The conventional method violates the safety constraint \(x\le1\) after the attack starts and becomes unstable, whereas the proposed method keeps the state stable, bounded and within the nominal safe set, as shown in the inset. The input responses further show that the compensated actuator input remains bounded and returns toward its nominal behavior after the attack is removed.  Fig.~\ref{fig:case1_proposed} considers the more aggressive attack \(d^{(2)}(t)\) and illustrates UUS in Definition~\ref{def:UUS}. The state exhibits only a bounded transient excursion into the enlarged safe set, then re-enters the nominal safe set in finite time and remains there. This shows that the proposed AR-CLF-CBF mechanism recovers nominal safety after a finite transient, while the adaptive gains \(\rho(t)\) and \(\eta(t)\) increase online to dominate the attack effects.
\begin{figure}[t]
\centering
\includegraphics[width=0.49\textwidth]{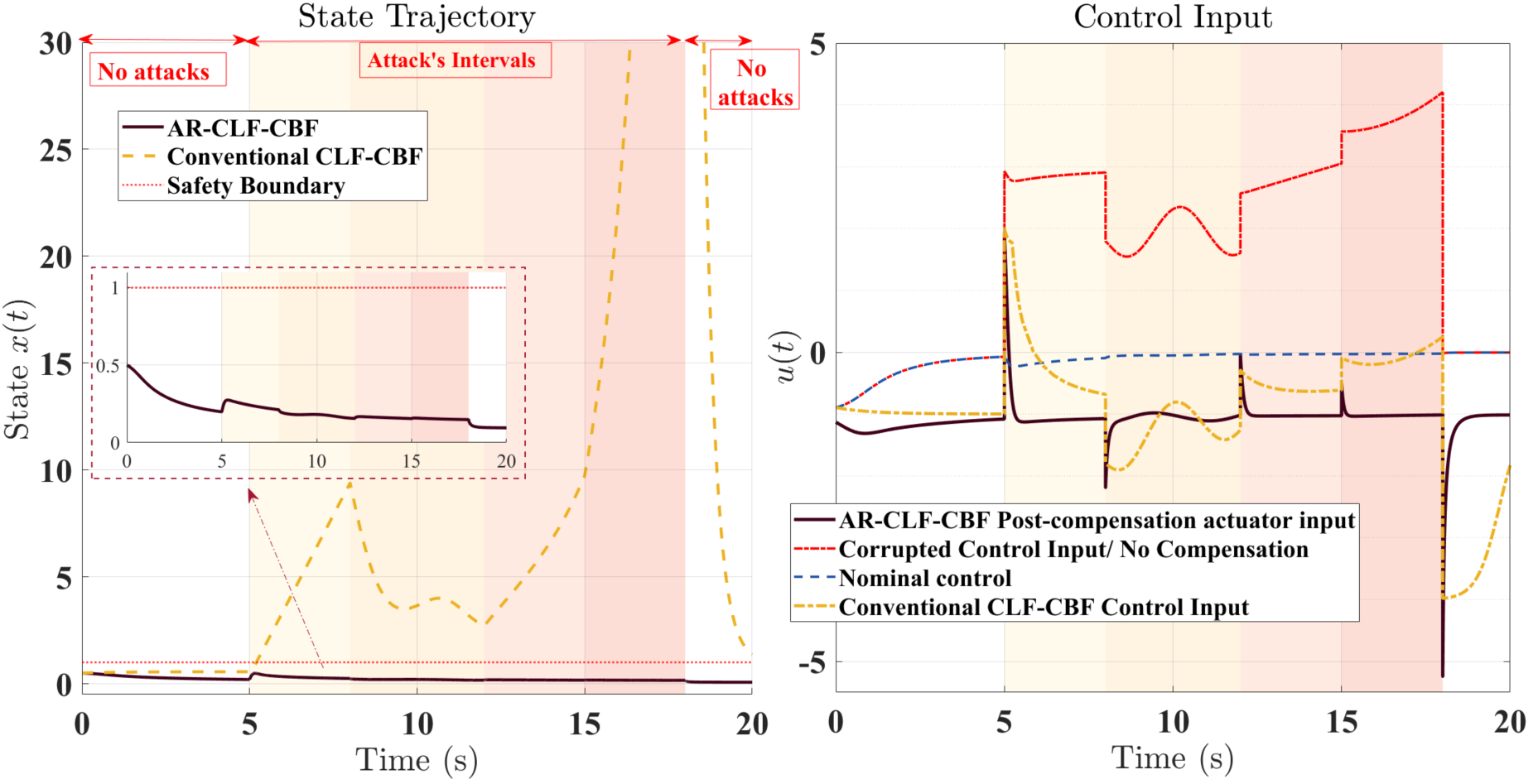}
\caption{Comparison between the proposed AR-CLF-CBF-QP and the conventional CLF-CBF-QP under the attack profile \(d^{(1)}(t)\).}
\label{fig:case2_compare}
\end{figure}
\begin{figure}[t]
\centering
\includegraphics[width=0.49\textwidth]{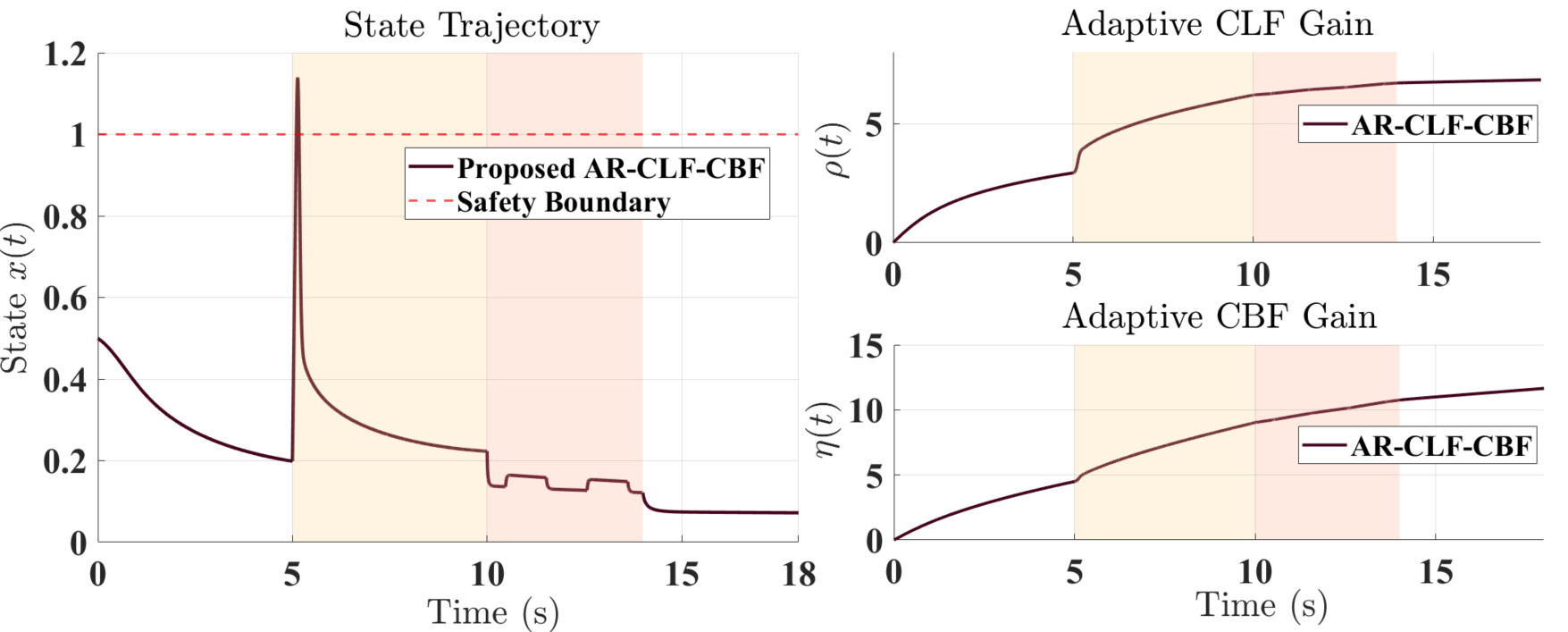}
\caption{Performance of the proposed AR-CLF-CBF-QP under the attack profile \(d^{(2)}(t)\).}
\label{fig:case1_proposed}
\end{figure}

\textbf{Example 2: 2-DOF Robot Under Input FDI Attacks:}
We validate the proposed AR--CLF--CBF-QP on the 2-DOF robot in
Fig.~\ref{Robot_Ex}, with revolute coordinate \(\theta\), prismatic coordinate
\(r\), input \(u=(\tau,T)\), and dynamics
\begin{equation}
D(q)\ddot q + C(q,\dot q) = u,
\qquad q = (\theta,r),
\label{eq:robot_dyn}
\end{equation}
\begin{equation*}
D(q)=
\begin{bmatrix}
m r^2 + \frac{M L^2}{3} & 0\\[2pt]
0 & m
\end{bmatrix},
\quad
C(q,\dot q)=
\begin{bmatrix}
2mr\dot r\dot\theta\\[2pt]
- mr\dot\theta^2
\end{bmatrix},
\end{equation*}
with \(m=M=1\) kg and \(L=3\) m. The objective is regulation to
\(q_d=(0,1.5)\) while enforcing \(r(t)\le r^\star=2\) m. We use the energy-based CLF $V(x)=(q-q_d)^\top K_p(q-q_d)+\dot q^\top D(q)\dot q$ where $K_p=K_d=I_2$,
with the AR--CLF constraint
\[
L_fV+L_gV\,u+
\frac{\|L_gV\|^2}{\|L_gV\|+e^{-\alpha t^2}}
e^{\rho(t)}
\le
-\dot q^\top K_d\dot q+\delta,
\]
Since the safety function \(h(x)=r^\star-r\) has relative degree two, we use
the AR--HOCBF constraint Following standard HOCBF constructions \cite{ames2014adaptive},
\[
L_f^2h+L_gL_fh\,u+
\frac{\|L_gL_fh\|^2}{\|L_gL_fh\|+e^{-\alpha t^2}}
e^{\eta(t)}
\ge
-k_p h-k_dL_fh,
\]
with \(k_p=1\) and \(k_d=1.73\). Then control input is obtained from the AR--CLF--CBF-QP with relaxation \(\delta\ge0\). We test two bounded and two unbounded FDIA on
the second input to compare different adversarial growth profiles. The ISSf-CBF method in~\cite{kolathaya2018input} is used as a disturbance-robust
baseline because it guarantees invariance of a disturbance-dependent enlarged safe set under bounded disturbances, whereas the proposed method aims to recover and maintain the nominal safe set under unbounded control-input FDIA. As shown in Fig.~\ref{Robot_Ex}, ISSf-CBF handles the bounded cases but fails under unbounded attacks, where \(r(t)\) violates the safety limit and grows rapidly. In contrast, the proposed AR--CLF--CBF-QP maintains \(r(t)\le r^\star\), even under quadratically growing attacks, by increasing the adaptive gains and generating corrective inputs online.

\begin{figure}[ht]
\centering
\begin{minipage}{0.1\textwidth}
    \centering
    \includegraphics[width=\linewidth]{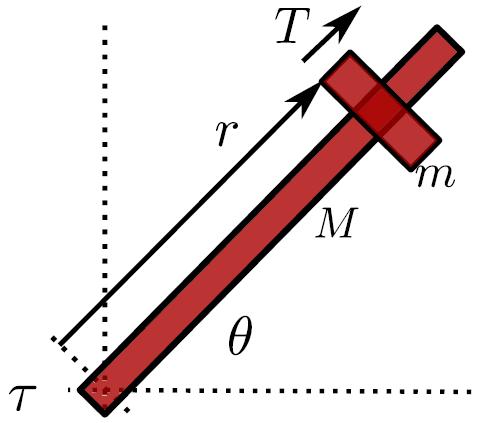}
\end{minipage}
\hspace{0.05mm}   
\begin{minipage}{0.36\textwidth}
    \centering
    \includegraphics[width=\linewidth]{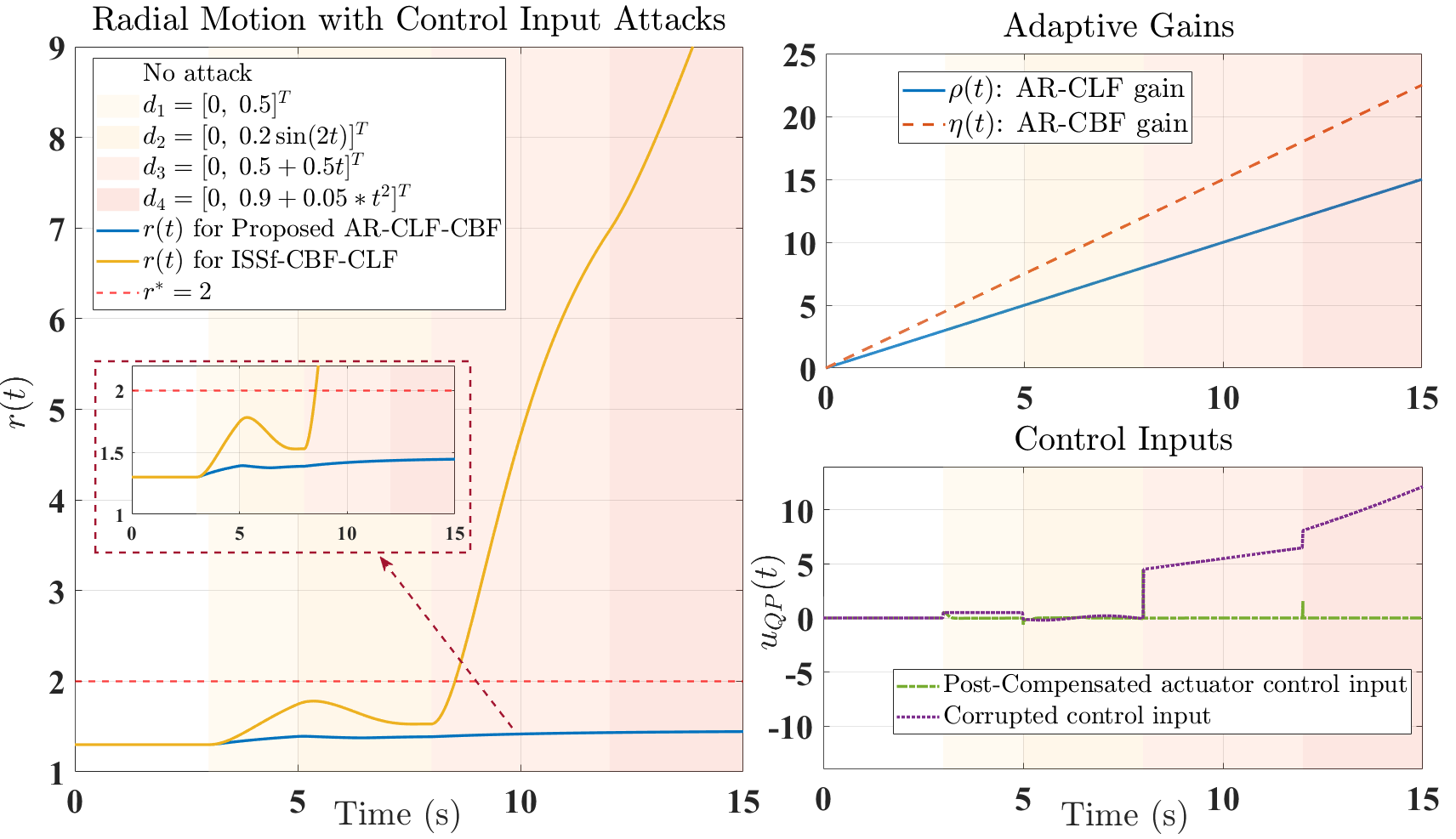}
\end{minipage}
\caption{(a) Robot structure and (b) Comparison between proposed AR-CLF-CBF with ISSf-CBF in \cite{kolathaya2018input}.}
\label{Robot_Ex}
\end{figure}
\section{Conclusion}
This paper introduced new notions of AR-CLFs and AR-CBFs, leading to a unified AR-CLF-CBF-QP framework for nonlinear affine systems under a broad class of adversarial  control-input FDIA, including time-varying and unbounded perturbations. By embedding a unified attack-compensation mechanism within both the CLF decrease condition and the CBF forward-invariance constraint, together with an online gain tuning law, the proposed framework enables real-time compensation of adversarial effects without requiring prior knowledge of attack bounds. Unlike ISS/ISSf methods that certify only disturbance-dependent enlarged safe sets, the proposed framework adaptively recovers the nominal safe set while jointly preserving stability and safety; Numerical results show improved resilience over robust CLF-CBF baselines.

\ifCLASSOPTIONcaptionsoff
  \newpage
\fi

\bibliographystyle{IEEEtran}

\bibliography{References}


\end{document}